\documentclass[12pt,onecolumn,journal,letterpaper,romanappendices]{IEEEtran}
\usepackage{epsfig,rotating,setspace,latexsym,amsmath,epsf,amssymb,bm}
\usepackage{color}

\usepackage{amsmath}
\usepackage{rotating,latexsym,amsmath,amssymb,bm}
\usepackage{cite}
\usepackage{amsthm}
\usepackage{subfigure}
\usepackage{graphicx}
\def\nba{{\mathbf{a}}}

\def\nbs{{\mathbf{s}}}

\def\nbw{{\mathbf{w}}}
\def\nbx{{\mathbf{x}}}
\def\nby{{\mathbf{y}}}

\def\nb0{{\mathbf{0}}}
\def\nb1{{\mathbf{1}}}

\def\nbX{{\mathbf{X}}}
\def\nbY{{\mathbf{Y}}}
\def\nbZ{{\mathbf{Z}}}

\def\ncalA{{\mathcal{A}}}
\def\ncalB{{\mathcal{B}}}

\def\ncalD{{\mathcal{D}}}

\def\ncalF{{\mathcal{F}}}

\def\ncalI{{\mathcal{I}}}

\def\ncalM{{\mathcal{M}}}

\def\ncalS{{\mathcal{S}}}

\def\ncalW{{\mathcal{W}}}

\def\ncalY{{\mathcal{Y}}}

\def\nbbE{{\mathbb{E}}}

\def\nbbN{{\mathbb{N}}}

\def\nbbP{{\mathbb{P}}}

\def\nbbR{{\mathbb{R}}}

\newtheorem{theorem}{Theorem}

\def\E{\mathbb{E}}

\def\sinr{\mathtt{SINR}}			

\usepackage{bbm}
\usepackage{url}

\usepackage{wrapfig}
\usepackage{dsfont}
\usepackage{wasysym}
\usepackage{hyperref}
\usepackage{balance}
\usepackage[ruled,vlined]{algorithm2e}

\theoremstyle{plain}

\newtheorem{lemma}{Lemma}

\newtheorem{prop}{Proposition}

\theoremstyle{definition}
\newtheorem{Def}{Definition}
\newtheorem{assum}{Assumption}

\usepackage{thmtools}

\declaretheoremstyle[
  spaceabove=\topsep, spacebelow=\topsep,
  headfont=\normalfont\bfseries,
  notefont=\mdseries, notebraces={(}{)},
  bodyfont=\normalfont,
  postheadspace=1em,
  qed=\qedsymbol
]{mythmstyle}

\declaretheoremstyle[
  spaceabove=\topsep, spacebelow=\topsep,
  headfont=\normalfont\bfseries,
  notefont=\mdseries, notebraces={(}{)},
  bodyfont=\normalfont,
  postheadspace=1em,
  qed=$\diamond$
]{mythmstyle}
\declaretheorem[style=mythmstyle]{remark}

\DeclareMathOperator{\vE}{\varepsilon}

\DeclareMathOperator{\kl}{\operatorname{D}_{\texttt{KL}}}

\begin{document}

\allowdisplaybreaks

\sloppy

\title{Universal Learning Waveform Selection Strategies for Adaptive Target Tracking}

\author{Charles E. Thornton$^{*}$, R. Michael Buehrer$^{*}$, Harpreet S. Dhillon$^{*}$,\\ and Anthony F. Martone$^{\dagger}$
\thanks{$^{*}$Bradley Department of ECE, Virginia Tech, Blacksburg, VA, USA, 24061. Correspondence: $thorntonc@vt.edu$. $^{\dagger}$ US Army Research Laboratory, Adelphi, MD, USA, 20783. A preliminary version of this work was presented at MILCOM 2021, San Diego, CA \cite{Thornton2021}. Support from the US Army Research Office (ARO) made this work possible.}}

\maketitle
\thispagestyle{plain}
\pagestyle{plain}
\vspace{-1cm}
\begin{abstract}
Online selection of optimal waveforms for target tracking with active sensors has long been a problem of interest. Many conventional solutions utilize an estimation-theoretic interpretation, in which a waveform-specific Cram\'{e}r-Rao lower bound on measurement error is used to select the optimal waveform for each tracking step. However, this approach is only valid in the high SNR regime, and requires a rather restrictive set of assumptions regarding the target motion and measurement models. Further, due to computational concerns, many traditional approaches are limited to near-term, or \emph{myopic}, optimization, even though radar scenes exhibit strong temporal correlation. More recently, reinforcement learning has been proposed for waveform selection, in which the problem is framed as a Markov decision process (MDP), allowing for long-term planning. However, a major limitation of reinforcement learning is that the memory length of the underlying Markov process is often unknown for realistic target and channel dynamics, and a more general framework is desirable. This work develops a universal sequential waveform selection scheme which asymptotically achieves Bellman optimality in any radar scene which can be modeled as a $U^{\text{th}}$ order Markov process for a finite, but unknown, integer $U$. Our approach is based on well-established tools from the field of universal source coding, where a stationary source is parsed into variable length phrases in order to build a \emph{context-tree}, which is used as a probabalistic model for the scene's behavior. We show that an algorithm based on a multi-alphabet version of the Context-Tree Weighting (CTW) method can be used to optimally solve a broad class of waveform-agile tracking problems while making minimal assumptions about the environment's behavior. 
\end{abstract}


\section{Introduction and Prior Work}\label{sec:intro}

\subsection{Waveform-Agile Tracking}
In active tracking systems, such as radar and sonar, the choice of transmitted waveform plays a significant role in determining both the nature and quality of target information obtained at the receiver \cite{Woodward1953}. The essential properties of a waveform are identified through the \emph{Woodward ambiguity function} (AF), which specifies the matched-filter response to the given waveform over a two-dimensional space of delay and Doppler values \cite[Ch. 3]{Levanon2004}. There is unfortunately no general procedure for synthesizing a signal directly from the AF. Traditionally, appropriate waveforms have been selected by human experts such that the expected matched-filter response is appropriate for a specific application \cite{Aasen1976}. However, when a target is being tracked in a dynamic environment, the system's need for a particular type of measurement, such as in the range or Doppler dimension, will vary depending on instantaneous uncertainty. For example, the tracker's need may depend on the channel's state, the target's state, and the sequence of previously recorded measurements. In many dynamic scenarios, it is difficult to predict the resolution requirements of a tracking system \emph{a priori}, and on-the-fly waveform selection schemes are necessary.\\ 

Appreciable performance benefits have been demonstrated when the AF of the transmitted waveform is dynamically tailored to the state of the sensing environment \cite[and references therein]{Sira2009}. It is also worth noting parallels between waveform-agile radar tracking and the echolocation strategy employed by bats, which use a variety of frequency modulated and constant frequency pulses when hunting a target, instead of simply using a waveform with a large time-bandwidth product, as one might expect \cite{Simmons1992}.\\

Recognizing the need for dynamically tailored waveforms, two intimately related questions arise:\\
\begin{enumerate}
	\centering
	\item ``Which waveform is optimal for a particular sensing environment?" \newline \emph{(Waveform design, studied in \cite{Bell1988,Bell1993,Hong2005,Sira2007,Sira2007a,Zhu2017,Kay2009,Kay2007,DeLong1969} et al.)}
	\item ``How should a real-time system select appropriate waveforms using sequential observations?" \emph{(Waveform selection, studied in \cite{Kershaw1994,Kershaw1997,Niu2002,LaScala2005,PLiu2020,Thornton2020} et al.)} \vspace{.4cm}
\end{enumerate}

In this paper, we focus on the latter question, while touching on notions of waveform optimality from the former. In particular, we pose the waveform selection problem as a Markov decision problem (MDP) of finite, but unknown order. This formulation is very general, and few assumptions about the statistical behavior of the underlying radar scene are made. To find an optimal transmission policy, we develop a novel reinforcement learning algorithm based on the Context-Tree Weighting (CTW) approach to universal data compression, which is well-known to provide both strong theoretical guarantees and practical performance in both source coding and sequential prediction problems \cite{Willems1995,Willems1998,Begleiter2004}. We show that this algorithm converges to a Bellman-optimal policy asymptotically\footnote{Bellman-optimal policies are not guaranteed to be unique} for a very broad class of scenes while making few assumptions about the statistical behavior of the scene itself. We also demonstrate that an implementation of this approach results in performance improvements over traditional adaptive waveform selection strategies, and results in favorable finite-time behavior compared with the Lempel-Ziv inspired universal reinforcement learning scheme proposed in \cite{Farias2010}. To the best of our knowledge, this is a new intersection between radar waveform selection and universal data compression, which presents opportunities for future work.

\subsection{Information-Theoretic Interpretation of Radar Signal Design and Selection}
A natural perspective from which to analyze radar signals is that of information theory, an approach first taken by Woodward and Davies in the 1950s \cite{Woodward1953,Woodward1951}. Although radar is fundamentally a measurement system, rather than one of communication, it is intuitive to frame the problem of transmission and reception in radar systems from the perspective of rate distortion theory, as the goal is to represent an unknown target process as efficiently as possible. \\

The information theoretic view of waveform design can be briefly summarized as follows. Consider a random vector of target parameters, $\mathbf{Z}$, a known transmitted waveform $\nbX$, and a random received signal $\nbY$. The waveform design problem consists of selecting $\nbX$ such that the mutual information $\operatorname{I}(\nbY;\nbZ|\nbX)$ is maximized. This Shannon-theoretic interpretation is adopted by Bell in his seminal work, which treats the two-way propagation of a radar transmission as an extension of the conventional communication channel \cite{Bell1993,Bell1988}. The power spectrum of the transmitted waveform is then matched to the target channel's impulse response using a ``water-filling" strategy, to ensure that mutual information between the target and received signal is maximized.\\

From a mathematical standpoint, the benefits of adaptive waveform selection are of little surprise, given well-established results in communication theory which document that an information source and channel must be in some sense \emph{probabalistically matched} to achieve optimal transmission \cite{Gastpar2003,Tatikonda2009}. While Bell's work provides very general bounds on the estimation capabilities associated with specific waveform-filter pairs, it tells us little as to how these bounds can be achieved in a practical system, where the target and channel statistics are unknown \emph{a priori} and the radar has the ability to adapt based on limited feedback. As a result, a large number of contributions have favored a statistical signal processing interpretation, selecting waveforms based on the criteria of expected Fisher information.

\subsection{Statistical Signal Processing Interpretation of Radar Signal Design and Selection}
In the signal processing literature, estimation-theoretic approaches for waveform selection based on the principle of Fisher information have been widely considered \cite{Kershaw1994,Kershaw1997,Sira2007,Hong2005}. Under this framework, a waveform-specific Cram\'er-Rao lower bound (CRLB) on target measurement errors can be obtained from the curvature of the ambiguity function at the origin of delay-Doppler space \cite{trees1992detection}, providing a strong theoretical basis for discrimination between waveforms. Unfortunately, the CRLB is only achieved at high $\texttt{SNR}$ values, which are rarely encountered in practical radar scenarios \cite{Sira2009}. \\

Additionally, CRLB-based approaches do not consider AF sidelobes, which may have a substantial impact on measurement quality, especially when perfect target detection cannot be assured. In the case of imperfect detection, the tracking error may be most efficiently reduced by taking measures to ensure reliable detections. For example, in \cite{Sira2007a} a two-stage procedure is used to gather information about the environment's clutter statistics and then perform online waveform design accordingly. However, only detection performance is considered and the procedure does not consider Doppler processing. The work in \cite{Niu2002} demonstrates improved tracking performance when the covariance matrix of the measurement error is calculated based on the AF resolution cells, which are regions of the delay-Doppler response which tightly enclose the AF shape of the transmitted waveform. This technique is suitable for low-SNR scenarios, since AF sidelobes are included in each cell. However, such a sampling grid can be difficult to establish for many waveforms of practical interest, and the waveform-dependent measurement noise covariance matrix must be computable in closed form. 

\subsection{Machine Learning Approaches to Waveform Selection}
A significant limitation of previous analytic solutions is that due to computational constraints, the tracker is usually limited to one-step ahead, or \emph{myopic}, optimization. However, the sensing environment for tracking radar usually has extended temporal correlations. This is especially relevant when the radar is interference or clutter-limited \cite{Ward1997,Siegal1994,Rosenberg2012}. Thus, \emph{planning} future waveforms based on the scene's transition behavior becomes an important consideration \cite{Charlish2015}. It has been shown that sea clutter exhibits both temporal and spatial correlations \cite{Ward1990,Rosenberg2012}. Moreover, interference in dense wireless networks also exhibits temporal correlations, both in ad-hoc \cite{Ganti2009} and cellular networks \cite{Krishnan2017}. Interference is often a serious performance bottleneck for tracking radar, which operates in a low $\sinr$ regime, and spectrum sharing scenarios are expected to become more widespread in coming years \cite{Griffiths2015}. Thus, the need to address arbitrary temporal correlation is likely to grow over time. The demand for memory-based information processing has been further supported by a surge of research interest in so-called \emph{cognitive radar} systems. In \cite{Haykin2012}, Haykin argues that multi-scale memory is a key component of any radar system that might be considered `cognitive'. Haykin suggests using the \emph{encoding-decoding} principle, where memory is implemented using a fixed length encoder and decoder. In the present work, we suggest extending this principle to \emph{variable length coding}. \\

In recent years, several authors have applied machine learning techniques to the waveform selection problem. In particular, reinforcement learning has been extensively proposed, due to the ease with which feedback can be obtained from radar transmissions and the natural sequential structure of the waveform selection and related sensor management problems \cite{Thornton2020,Thornton2021b,Selvi2020,PLiu2020,Charlish2015,Krishnamurthy2002,Krishnamurthy2009}. Further, it is intuitive to consider a Markov channel model for the target tracking problem, as most tracking systems implicitly make a Markov assumption about the target's behavior \cite{Ristic2004}. Unfortunately, a major limitation in these formulations is ensuring that the scene is properly modeled as a MDP, for which the memory of the target channel must be known. While recurrent neural networks have been proposed as a means of addressing extended temporal correlations \cite{Thornton2020}, little can be said about the optimality or robustness of such an approach. To make fewer assumptions about the scene, it is of interest to design a system which is agnostic toward the memory length of the target channel. \\

\begin{figure*}[t]
	\centering
	\includegraphics[scale=0.8]{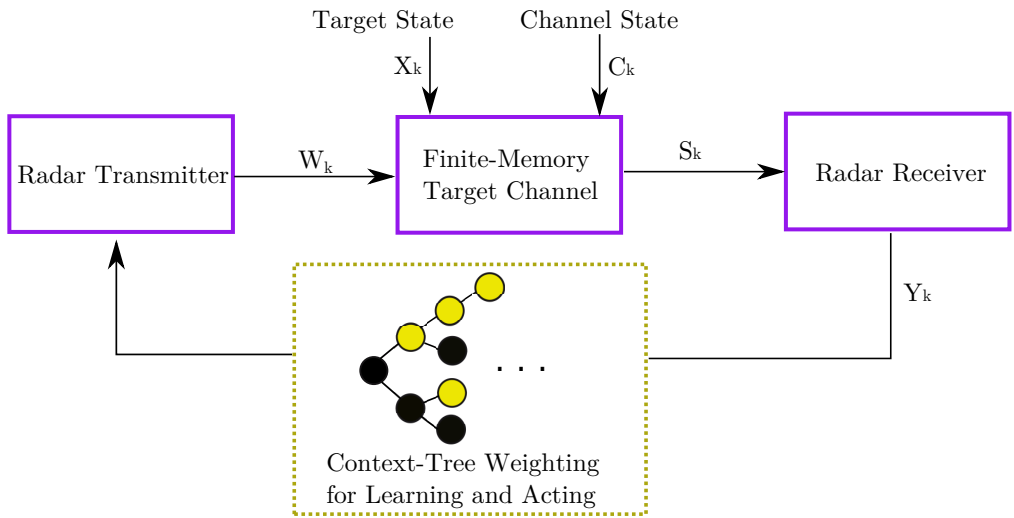}
	\caption{\textsc{Block diagram} of the proposed closed-loop radar system. The radar learns a context-tree model to caputre the unknown dynamics of the finite memory, finite state, target channel and sequentially selects waveforms accordingly.}
	\label{fig:block}
\end{figure*}

The current work builds on the paradigm of multi-scale memory, and addresses issues with previous reinforcement learning-based waveform selection approaches by applying sequence processing techniques from universal source coding \cite{Ziv1977} and universal prediction \cite{Merhav1998}. In particular, we consider the target channel to be a suffix-tree information source. We then apply a CTW-inspired algorithm, which builds a probabilistic model of a stationary tree source without a known distribution by parsing a sequence of states and actions into phrases of variable order. The CTW algorithm is powerful as it performs model weighting over the class of bounded memory tree sources. Thus, the algorithm performs well not only on average, but also on individual sequences. A visualization of our proposed scheme in the context of a general closed-loop radar system can be seen in Figure \ref{fig:block}.\\

Our proposed algorithm efficiently estimates both the transition probabilities and memory length of an unknown target channel by applying a \emph{variable length} encoder. The sequence of target channel observations and radar waveforms selected is used to build a context-tree that describes the evolution of a $U^{\text{th}}$ order MDP, where $U$ is an unknown fixed integer. We show that the proposed algorithm estimates the transition probabilities in an asymptotically optimal manner, and thus reaches an optimal waveform selection policy for the general model class considered. We also demonstrate a notable finite-time improvement over a universal reinforcement learning algorithm based on the Lempel-Ziv compressor that was proposed for a general setting in \cite{Farias2010}. 

\subsection{Our Contributions}
This work develops an online waveform selection algorithm which remedies many of the aforementioned issues with traditional waveform selection schemes. Our approach frames the online waveform selection problem for target tracking as a partially observable Markov Decision Problem with an unknown, but finite, memory window. This general formulation allows the approach to scale to tracking scenarios where the waveform-specific CRLB cannot be calculated or analytical assumptions are violated such that numerical approximations to the CRLB are no longer meaningful. Further, the proposed approach does not involve approximating clutter or interference distributions directly. While radar waveform selection problems have been framed as MDPs in the past, our work provides the following notable improvements:\\
\begin{itemize}
	\item We develop a novel online waveform selection algorithm based on an adaptation of Context-Tree Weighting to the reinforcement learning problem. We show that the proposed algorithm converges to the Bellman-optimal policy for any finite-state target channel with finite memory using results from bounded parameter MDPs and the uniformly bounded parameter redundancy of CTW. Further, we demonstrate that this algorithm provides favorable finite-time performance compared to the active Lempel-Ziv algorithm proposed for a general setting in \cite{Farias2010}. This is due to the fact that our algorithm inherits the optimal worst-case average redundancy performance of the CTW algorithm. \\ 
	
	\item We demonstrate through a simulation study that this approach results in favorable performance when the scene adapts to the radars behavior. We further present a modification of the algorithm which ignores the arbitrary dependence on previous states. We show that this variation greatly reduces the number of sample paths in the model and results in faster convergence in practical scenarios where the scene's evolution does not depend on the radar's choice of waveform. \\
		
	\item We directly include feedback from the tracking system in our model. While previous works have optimized task-based metrics such as spectral occupancy and $\texttt{SINR}$ \cite{Thornton2020, PLiu2020}, or have relied on numerical approximations for radar performance \cite{LaScala2005}, our approach optimizes tracking performance using the tracking filter's innovation.
\end{itemize}

\subsection{Notation}
Calligraphic symbols ($\mathcal{A},\mathcal{S})$ denote sets, with $n$-letter extensions $(\ncalA^{n},\ncalS^{n})$. Capital letters ($X$,$Y$) denote random variables. Uppercase and lowercase bold symbols ($\mathbf{H},\mathbf{h})$ denote matrices and vectors, respectively. $\nbbR_{+}$ and $\nbbN_{+}$ denote the set of positive real and natural numbers. $\nbbR^{n}$ denotes $n$-dimensional Euclidean space. $\nbbE[\cdot]$ is the expectation of a random variable. To denote a sequence $\{\nbx_{1},\nbx_{2},...\nbx_{N}\}$, we use the shorthand $\nbx_{1}^{N}$, occasionally using $\nbx_{1}^{N}$ and $\nbx^{N}$ interchangeably. $\nbbP(\cdot)$ denotes a probability measure over a measurable space $(\Omega,\ncalF)$.\\

\section{Problem Formulation}
\subsection{System Model}
Consider a single-target moving in $p$-dimensional space. The target is being tracked by a waveform-agile radar sensor, which is stationary and located at the origin. The radar has access to a finite alphabet of waveforms, denoted by the set $\{\boldsymbol{\theta}_{i}\} = \Theta$. Let $\{ k \in \mathbb{N}_{+}: 1 \leq k \leq K \} = [K]$ denote the sequence of tracker update steps. During each step, a train of $N_{p}$ pulses is transmitted. Multi-pulse processing is used to ensure good resolution in both the range and Doppler directions \cite{Levanon2004}. The state of the target at step $k$ is denoted by $\mathbf{X}_{k} \in \mathbb{R}^{d}$, which is a vector containing $d$ parameters the radar wishes to estimate. Each step, the radar does not observe the true state, but instead receives a noisy measurement, denoted by $\mathbf{Z}_{k} \in \mathbb{R}^{d}$. \\

We consider both intrapulse (signal) processing and interpulse (information) processing in the tracking formulation. Each tracking step $k$, which consists of $N_{p}$ coherent pulses, the sensor transmits a pulse train given by the model \cite[p. 240]{VanTrees2001}
\begin{equation}
	s_{T}(t) = \sqrt{2} \textrm{Re} [ \sqrt{E_{T}} \tilde{s}(t) e^{j 2 \pi f_{c} t} ],
\end{equation}
where $t$ corresponds to continuous `fast time', $E_{T}$ is the energy of the transmitted signal, $f_{c}$ is the carrier frequency, and $\tilde{s}(t)$ is the complex envelope, which is assumed to be normalized $\int_{-\infty}^{\infty} |\tilde{s}(t)|^{2} dt = 1$ and is given by
\begin{equation}
	\tilde{s}(t) = \frac{1}{\sqrt{N_{p}}}\sum_{n=1}^{N_{p}-1} \tilde{s}_{n}(t- (n-1) T_{r}),
\end{equation}
where each $\tilde{s}_{n}$ is the complex envelope of an individual pulse and $T_{r}$ is the \emph{pulse repetition interval}. The measurement at each step then depends on the received signal\footnote{We study the case of a point target for simplicity. However, the waveform selection scheme presented here is also applicable to extended targets \cite[Ch. 12]{VanTrees2001}.}, which is given by the model 
\begin{equation}
	s_{R}(t) = \sqrt{2} \textrm{Re}\{ [ \sqrt{E_{R}} e^{j \phi} \tilde{s}(t-\tau_{0})e^{j \nu_{0}t} + \tilde{n}(t) ] e^{j 2 \pi f_{c} t}\},
\end{equation}
where $E_{R}$ is the energy of the received pulse, $\phi$ is a random phase shift, and $\tilde{n}(t)$ is a zero-mean complex noise process, often assumed to be white and Gaussian. In this work, we will avoid making any explicit assumptions about the noise, clutter or interference processes. The parameters $\tau_{0}$ and $\nu_{0}$ are the delay and Doppler shift, respectively. Our signal model assumes the maximum Doppler shift is small compared to the carrier frequency, which is true for most radar systems since the propagation velocity is approximately the speed of light. \\

To extract an estimate of the delay and Doppler values $\tau_{0}$ and $\nu_{0}$ from the received pulse train, a bank of matched filters is applied over a range of possible delay and Doppler values. The normalized response of this processor is given by the periodic ambiguity function
\begin{equation*}
	|A_{s}(\tau,\nu)| = \left| \frac{1}{N_{p}T_{r}} \int_{0}^{N T_{r}} \tilde{s}(t) \tilde{s}^{*}(t+\tau) \exp{(j2 \pi \nu t)} \; dt \right|,
\end{equation*}
which describes the fundamental resolution properties of the pulse train in the delay-Doppler space of interest. 

\subsection{Waveform Dependent Fisher Information (Traditional Solution)}
In this section, we describe the conventional strategy for waveform selection, which is based on Fisher information \cite{Kershaw1994,Kershaw1997,Sira2007,Sira2009} and discuss some practical limitations. Under the assumption of Gaussian measurement errors and a sufficiently high SNR, the AF of the transmitted waveform can be used to find the Fisher information matrix (FIM), which is given by
\begin{equation}
	\label{eq:fim}
	\mathbf{I}=\eta\left[\begin{array}{cc}
		\bar{f}^{2}-(\bar{f})^{2} & \overline{f t}-\bar{f} \bar{t} \\
		\overline{f t}-\bar{f} \bar{t} & \bar{t}^{2}-(\bar{t})^{2}
	\end{array}\right],
\end{equation}
where $\bar{f}$ and $\bar{t}$ are the mean frequency and time, respectively, and $\eta$ is the SNR. The elements of the information matrix are related to the second derivatives of the AF evaluated at the true target position \cite[p. 297-300]{VanTrees2001}. For differentiable waveforms, we can thus write the FIM in terms of the waveform parameters $\boldsymbol{\theta} \in \Theta$. However, closed form expressions can only be found for certain classes of waveforms, as is noted in \cite{Kershaw1994,Kershaw1997,Sira2009}. The FIM, and associated CRLB, can be further integrated into a tracking system and used as a criterion for waveform selection. However, this approach is dependent on a high SNR, and tracking is most often performed in low SNR scenarios. Further, in interference or clutter-limited scenarios, the assumption of a Gaussian noise model is often unfounded. To relax these assumptions, the remainder of this work investigates a more general formulation, in which a variable-order Markov model is used to select waveforms based on the history of radar observations.

\subsection{Target Motion and Measurement Model}
The interpulse target state and and observation dynamics can be generally formulated as follows
\begin{equation}
	\label{eq:generalMotion}
\begin{split}
	\mathbf{X}_{k+1} &= f_{k}(\mathbf{X}_{k},\mathbf{W}_{k}), \quad \text{and} \\
	\mathbf{Z}_{k} &= h_{k}(\mathbf{X}_{k},\mathbf{V}_{k}),
\end{split}
\end{equation}
where $f_{k}$ and $h_{k}$ are time-varying functions that may be non-linear, and $\mathbf{W}_{k}$ and $\mathbf{V}_{k}$ are the state and measurement noise processes, respectively. \emph{In this work, we do not make any specific assumptions regarding the statistics of the noise processes}. The goal of a tracker is to sequentially estimate the probability density function of the state using the observed measurements, $\mathbb{P}(\mathbf{X}_{k}|\mathbf{Z}_{1}^{k})$. Then the associated mean $\hat{\mathbf{X}}_{k} = \E[\mathbf{X}_{k}|\mathbf{Z}_{1}^{k}]$ can be used as an estimate of the target's state. \\

Notably, the quality of the recieved measurements depends on the resolution properties of the transmitted waveform, which is exactly makes adaptive waveform selection interesting. Let the vector $\boldsymbol{\theta}_{k} \in \Theta$ denote the choice of waveform shape and the associated parameters transmitted at step $k$. The tracker's pdf of interest is then $\mathbb{P}(\mathbf{X}_{k}|\mathbf{Z}_{1}^{k}, \boldsymbol{\theta}_{1}^{k})$ with associated mean $\hat{\mathbf{X}}_{k} = \E[\mathbf{X}_{k}|\mathbf{Z}_{1}^{k},\boldsymbol{\theta}_{1}^{k}]$. A reasonable goal for tracking is to minimize the mean square error. Thus, we can establish the following waveform-dependent cost function
\begin{equation}
	J(\boldsymbol{\theta}_{k}) = \E_{\mathbf{X}_{k}, \mathbf{Z}_{k} \mid \mathbf{Z}_{1: k-1}, \boldsymbol{\theta}_{1:k-1}}\left(\left(\mathbf{X}_{k}-\hat{\mathbf{X}}_{k}\right)^{\mathrm{T}} \Lambda\left(\mathbf{X}_{k}-\hat{\mathbf{X}}_{k}\right)\right),
	\label{eq:idealcost}
\end{equation}
where the expectation is taken over possible states and measurements, and $\Lambda$ is a scaling matrix to ensure units remain consistent. To evaluate (\ref{eq:idealcost}), the waveform-specific measurement noise covariance matrix $N(\boldsymbol{\theta_{k}})$ must be computed for every possible waveform configuration. In the high SNR regime, it can be assumed that the radar reliably detects the target and reaches the CRLB. Thus $N(\boldsymbol{\theta}_{k})$ can be set using the CRLB as
\begin{equation*}
	N(\boldsymbol{\theta}_{k}) = \mathbf{T} \mathbf{I}^{-1}(\boldsymbol{\theta}_{k}) \mathbf{T}^{T},
\end{equation*}
where $\mathbf{I}(\boldsymbol{\theta})$ is the Fisher information matrix for the receiver and $\mathbf{T}$ is a transformation matrix between the receiver estimation parameters and the tracking subsystem measurement vector components.\\

In general, (\ref{eq:idealcost}) cannot be found in closed form, especially if the impact of a particular waveform choice on detection is unknown. If transmitting a particular waveform causes the radar to miss a detection, the resulting measurement will likely have a much larger tracking error than predicted by (\ref{eq:idealcost}). Previous works have relied on analytical approximations and sequential Monte-Carlo methods to approximate the one-step predicted MSE cost for every possible choice of waveform \cite{Sira2007,Sira2009,Bell2015}. However, both of these approaches come with clear disadvantages. For example, \cite{Bell2015} uses the predicted conditional CRLB to bound the MSE matrix by conditioning on received data. This approach relies on approximations in many cases, with Gaussian distributed errors being a notable exception. Additionally, these approaches are generally limited to myopic optimization due to computational concerns, where a multi-step optimization is expected to be more effective for many tracking applications. 

\subsection{Traditional Waveform Selection}
If we constrain the formulation of (\ref{eq:generalMotion}) and instead consider the discrete-time linear target and measurement models
\begin{align*}
	\mathbf{X}_{k+1} &= \mathbf{F} \mathbf{X}_{k} + \mathbf{G} \mathbf{W}_{k} \\
	\mathbf{Z}_{k}   &= \mathbf{H} \mathbf{X}_{k} + \mathbf{V}_{k},
\end{align*}
where $\mathbf{W}_{k}$ and $\mathbf{V}_{k}$ are zero-mean, white Gaussian noise processes, which are both zero-mean with respective covariance matrices $\mathbf{Q}_{k}$ and $\mathbf{N}_{k}$, then a closed-form expression for the waveform-dependent measurement noise covariance matrix $\mathbf{N}(\boldsymbol{\theta_{k}})$ can be found using the FIM $\mathbf{I}(\boldsymbol{\theta_{k}})$ written in (\ref{eq:fim}). However, as discussed previously, we wish to avoid any particular assumptions regarding the statistics of the target channel. Thus, instead of solving for $\mathbf{N}(\boldsymbol{\theta_{k}})$ in order to directly estimate the impact of each parameter vector $\boldsymbol{\theta_{k}} \in \Theta$ on tracking performance, we more generally describe the radar waveform selection process as an optimal control problem in a finite state target channel. Considerations related to the evolution of the target channel and the radar's measurement processes are now described.

\section{Universal Reinforcement Learning Problem}
Consider a radar system located at the origin. We assume the radar is stationary, and continually surveys a scene which contains at most one moving target of interest. The scene may contain other scatterers (clutter) as well as interfering signals. The scene can be represented as a grid in the delay-Doppler domain, where the range of quantized delay cells is indexed by $\tau = \{1,...,M\}$ and the collection of Doppler cells is indexed by $\nu = \{1,...,N\}$. Let the space of hypotheses about the target's position be denoted by set $\mathcal{H}$, where $|\mathcal{H}| = (MN)+1$, with the additional hypothesis corresponding to no target present in the scene. \\

During each tracking step, the radar makes an observation $y_{k}$ from a finite alphabet $\ncalY$, must select a waveform $\mathbf{w}_{k}$ from finite alphabet $\mathcal{W}$, and receives a bounded cost $g(\nby_{k},\nbw_{k},\nby_{k+1}) \in [-g_{max},g_{max}]$. The radar wishes to minimize the long-term average cost given by 
\begin{equation}
	\limsup_{K \rightarrow \infty} \E \left[ \frac{1}{K} \sum_{k=1}^{K} g(\nby_{k},\nbw_{k},\nby_{k+1})  \right],
\end{equation} 
where the difficulty comes from the apparent randomness in the observation process, and the lack of prior knowledge regarding the cost function. We now describe the Markov kernel which governs the measurement generating process.\\

The target's state, which corresponds to its true position, is denoted by $\mathbf{x}_{k} = i$, where the index $k \in \mathbb{N}_{+}$ corresponds to the discrete tracking step, and $i \in \mathcal{H}$ is a hypothesis regarding the target's location in the delay-Doppler space of interest. The target's state evolves according to a finite-memory stochastic process with stationary transition probabilities 
\begin{equation}
	P(\nbx_{k+1}|\mathbf{x}_{k-L+1}^{k}) = \mathbb{P} \left( \mathbf{x}_{k+1} = i | \mathbf{x}_{1}^{k} \right) \quad \forall i \in \mathcal{H}, \mathbf{x}_{k-L+1}^{k} \in \mathcal{H}^{L}
\end{equation}
where $P$ is a stochastic transition kernel and $L \in [0, \infty]$ is the memory length of the motion process, and the distribution exists for all possible values of the past $L$ target states $\mathbf{x}_{k-L+1}^{k}$. Both the target's state transition probabilities and the process memory length are assumed to be unknown to the radar \emph{a priori}. The radar's goal is to determine the target's state at each time step with minimal uncertainty. \\

In addition to the target's state, the scene is also characterized by the state of the \emph{target channel}, which corresponds to the relative degree of propagation losses due to interference and noise in the channel \cite{Bell1988}. The target channel's state is denoted by $\mathbf{c}_{k} \in \mathcal{C}$, where $|\mathcal{C}| < \infty$. Similar to the target's state, the state of the target channel also evolves according to a finite-memory stochastic process, given by 
\begin{equation}
	P( \mathbf{c}_{k+1} = i | \mathbf{c}_{k-J+1}^{k}, \mathbf{w}_{k-J+1}^{k}), \quad \forall i \in \mathcal{C}, \mathbf{w}_{k-J+1}^{k} \in \mathcal{W}^{J}, \mathbf{c}_{k-J+1}^{k} \in \mathcal{C}^{J}
\end{equation}
where $\mathbf{w}_{k-J+1}^{k}$ are the past $J$ transmitted waveforms, each selected from a finite alphabet, $|\mathcal{W}| < \infty$. Design choices for offline construction of this library will depend on several application specific considerations, and are described in Section \ref{se:library}. The fixed constant $J \in [0,\infty]$ is the memory length of the process.\\ 

It is important to note the general dependence of the previous $J$ transmitted waveforms on the evolution of future target channel states. Although many real-world interferers will not respond to the radar's choice of waveform, this model is general enough to consider a reactive interfering system, which may co-operate or compete with the radar for channel resources. This dependence can also be relaxed to simplify the problem considerably.\\

 The true state of the radar's scene can then be viewed as the composition $\mathbf{s}_{k} = [\mathbf{x}_{k}, \mathbf{c}_{k}]$ which takes values in set $\mathcal{S}$, which has cardinality $|\mathcal{H}| \times |\mathcal{C}|$. The scene transition probabilities can then be similarly expressed as 
\begin{equation}
	 \label{eq:stateTrans}
	 P(\mathbf{s}_{k+1}| \mathbf{s}_{k-U+1}^{k}, \mathbf{w}_{k-U+1}^{k}),
\end{equation}
where $U = \max \{J,L\}$ is the memory length of the state generating process.\\

In general, it cannot be assumed that the radar directly observes the composite state $\mathbf{s}_{k}$. Instead, each tracking step $k \in \mathbb{N}_{+}$, the radar receives a noisy measurement $\mathbf{y}_{k} = j$ from set $\mathcal{Y}$, where $|\mathcal{Y}| = 2^{NM} \times |\mathcal{C}|$. The probability of receiving a particular measurement $\mathbb{P}(\mathbf{y}_{k} = j)$ will depend on the most recent waveform transmitted by the radar and the state of the scene, given by $\mathbf{s}_{k} \in \mathcal{S}$. The probability of receiving measurement $\nby_{k} = j$ is then given by the following measurement model
\begin{equation}
	\mathbb{P} \left( \nby_{k+1}=j|\nbs_{k+1}=i, \nbw_{k} = h \right) \quad \forall j \in \mathcal{Y}, i \in \mathcal{S}, h \in \mathcal{W},
\end{equation}
which reflects uncertainty in the radar's observation of the target's state and the channel state due to estimation errors.\\ 

Since the radar does not directly observe $\nbs_{k}$, it must make a decision regarding $\nbw_{k}$ using only the information sequentially observed up until time $k$, given by the set
\begin{equation} 
\mathcal{F}_{k} = \{\nby_{1},\nbw_{1},g_{1},\nby_{2},...,\nbw_{k-1},g_{k-1},\nby_{k}\}, 
\end{equation}
which is the sequence of nested sub $\sigma$-algebras generated by the observations, costs and actions until time $k$. This set can equivalently be represented as the smallest $\sigma$-algebra generated by a sequence of random variables, ie. $\sigma(F_{1}^{k})$ where the random variables are defined on a common measurable space ($\Omega,\ncalF)$ and $\ncalF_{1} \subseteq \ncalF_{2} \subseteq ...\subseteq \ncalF$. This set is often referred to as the \emph{information state} or filtration \cite[Section 2.2]{Lattimore2020}. It is assumed that the radar can store the entire information state in memory to enable knowledge gain. \\

Since $\mathcal{F}_{k}$ contains all relevant information gathered until decision step $k$, it can be used to select waveforms in place of the true state sequence $\{\nbs_{1},...,\nbs_{k}\}$. The main difficulty with using the information state to select waveforms is that the dimension of $\mathcal{F}_{k}$ grows linearly with $k$. To keep the problem tractable, the radar can utilize a sufficient statistic for $\mathcal{F}_{k}$. In this work, we make the following assumption: \\

\begin{assum}
	Conditioned on the current information state, the noisy measurement $\nby_{k}$ is an unbiased estimate of the true state of the scene $\nbs_{k}$, ie. $\E[\nby_{k+1}|\mathcal{F}_{k}] = \E[\nbs_{k+1}|\mathcal{F}_{k}]$. Thus $\nby_{k}$ is a sufficient statistic for $\nbs_{k}$.
\end{assum}

\begin{remark}
	The above assumption is reasonable when the radar has already detected the target from an earlier scanning period and obtains a measurement at each time step. In this case, both the target and spectrum observations can viewed as unbiased estimates of their true values.
\end{remark}

To calculate the measurement probabilities in closed form, a model-based approach can be used to calculate the probability of detection and false alarm in each cell, as in \cite{LaScala2005}. However, we will not make any particular assumptions about the structure of the target scene or waveform catalog, and take the view that these probabilities must be \emph{learned} over time through repeated experience. \\

The sequence of transmitted waveforms in response to the observed sequence of measurements can be interpreted as a policy, or decision function.
\begin{Def}[Policy]
	A policy $\mu$ is a sequence of mappings $\{\mu_{k}\}$, from which at each time $k$, the map $\mu_{k}: \mathcal{Y}^{k} \times \mathcal{W}^{k-1} \mapsto \mathcal{W}$ determines the waveform transmitted at time $k$ given the history of measurements and previously transmitted waveforms up to time $k$ from the information state $\mathcal{F}_{k}$.
\end{Def}

After each waveform decision, the radar receives a bounded cost $g(\nby_{k},\nbw_{k},\nby_{k+1}) \in [-g_{max},g_{max}]$, which quantifies the effect of $\nbw_{k}$ on the uncertainty about the target's position\footnote{Considerations related to cost function design will be discussed in Section \ref{se:cost}.}. To evaluate the performance of policy $\mu$, we define the long-term average cost while selecting waveforms according to $\mu$ as
\begin{equation}
	\lambda_{\mu} \triangleq \lim_{K \rightarrow \infty} \E_{\mu} \left[ \frac{1}{K} \sum_{k=1}^{K}g(\nby_{k},\nbw_{k},\nby_{k+1}) \right].
\end{equation}

Since the underlying state space is assumed to be finite, the above limit always exists \cite{Bertsekas2006}. We can then define the optimal average cost over stationary policies by
\begin{equation}
	\lambda^{*} \triangleq \inf_{\xi} \limsup_{K \rightarrow \infty} \E_{\xi} \left[ \frac{1}{K} \sum_{k=1}^{K}g(\nby_{k},\nbw_{k},\nby_{k+1}) \right],
	\label{eq:avCost}
\end{equation}
where the infimum is taken over the set of all admissible policies. The radar can then aim to find a policy $\mu$ that attains $\lambda^{*}$, provided that an appropriate cost function, which allows the radar to localize the target in the delay-Doppler grid of interest, is selected. In this formulation, we have made the following assumption.\\

\begin{assum}
	The optimal average cost is independent of the initial state. More precisely
	\begin{equation*}
		\lambda^{*}(\nby^{U},\nbw^{U-1}) = \lambda^{*}, \quad \forall (\nby^{U},\nbw^{U-1}) \in \ncalY^{U} \times \ncalW^{U-1},
	\end{equation*}
	which is a relatively benign assumption that holds for any problem satisfying a `weak accessibility' condition \cite{Bertsekas2006}.
	\\
\end{assum}

If the order of the state generating process $U$ is known, as well as each of the transition probabilities, then an average-cost optimal stationary policy can be found\footnote{While the optimal policy can be found via dynamic programming techniques, for large problem spaces this is often not possible in finite time and approximate solution methods are necessary.} by solving the discounted Bellman equation
\begin{equation}
	J(\nby^{U},\nbw^{U-1}) = \min_{w_{U}} \sum_{w_{U+1}} \nbbP(\nby_{U+1}|\nby^{U},\nbw^{U}) \times [g(\nby_{U},\nbw_{U},\nby_{U+1}) + \gamma J(\nby_{2}^{U+1},\nbw_{2}^{U})],
	\label{eq:Bellman}
\end{equation}
for all pairs $(\nby^{U},\nbw^{U-1}) \in \ncalY^{U} \times \ncalW^{U-1}$. If the discount factor $\gamma$ is sufficiently close to 1, a solution to (\ref{eq:Bellman}) can be used to find an optimal stationary policy for the average-cost problem posed in (\ref{eq:avCost}).

\begin{prop}[Structure of the optimal policy, Theorem 1 \cite{Yang2005}]
	For a feedback-dependent state generating process of the form $\mathbb{P}(\nbs_{k+1} = i| \nbs_{k-U}^{k},\nbw_{k-U}^{k})$, an optimal policy $\mu^{*}$, which achieves the optimal long-term average cost $\lambda^{*}$ will be a Markov process of order $U$, eg. a function of the form $\mu_{k}(\nbw_{k}| \nbs_{k-L}^{k}, \nbw_{k-L}^{k})$ where each $\mu_{k}$ is a probability distribution over the set $\ncalW$.
\end{prop}  

However, to find a policy corresponding to the notion of Bellman optimality seen in (\ref{eq:Bellman}), both the model order and transition probabilities are known. In a radar tracking problem, it cannot be assumed that the environment's dynamics are known \emph{a priori}. Thus, to learn an optimal waveform selection scheme, we must estimate both the model order of the transition kernel and transition probabilities themselves online in an efficient manner. Fortunately, estimation of an unknown order Markov process has been well studied in the field of universal source coding\footnote{An expository summary of universal data compression can be found in \cite[Chapter 13]{Cover2006}.}. In the next section, we review some of these tools and develop a waveform selection algorithm that inherits asymptotic Bellman optimality from the properties of known algorithms.

\section{Universal Waveform Selection Approach}
To sequentially learn an optimal policy for waveform selection, the evolution of the underlying dynamic environment must be estimated, while simultaneously controlling the radar system online. This presents a challenging problem, since the environment is assumed to have an unknown dependence on the past, and may also respond adversarially to the radar's choice of waveform, as seen by the dependence on $\nbw$ in (\ref{eq:stateTrans}). Further, the model could quickly become intractable for environments with many possible trajectories of the state-action sample path. Finally, the radar must strike a balance between selecting waveforms that improve the quality of the learned model (\textit{exploration}) and waveforms which minimize the incurred costs (\textit{exploitation}).\\

To establish Bellman optimality, both the transition kernel $P$ and cost-to-go function $J$ must be known, which can be used to select actions in a manner similar to the value iteration approach \cite[Sec. 1.3]{Bertsekas2006}. Since we focus on computing these values online, our analysis rests on the accuracy of our algorithms' estimates $\hat{P}$ and $\hat{J}$, similar to that of \cite{Kearns2002}, which proves a ``Simulation Lemma" that states conditions for sufficient estimation of $P$ and $J$ to select optimal actions. However, our problem requires slightly more care, as the model order is also being estimated online. In general, the rate at which the transition probability estimates converge to their true values will limit the performance of the algorithm. Thus, the empirical performance for a given radar scene will depend on the specific state transition structure.\\

In order to efficiently represent our model while simultaneously learning the model order and transition behavior, we choose to represent the $U^{\text{th}}$ order Markov environment using a context-tree data structure, which is a convenient way to represent an arbitrary process with memory. 

\begin{Def}[Tree Source]
	A tree source generates a sequence $\{\nby_{i},\nbw_{i}\}_{i=1}^{\infty}$ of symbols in a finite alphabet $\ncalD = \ncalY \times \ncalW$. The statistical behavior of a \emph{finite memory} tree source can be described by a suffix set $\ncalA$, which is a collection of strings on $\ncalD$ assumed to be proper and complete. Properness means that no string in $\ncalA$ is a suffix of any other string. Completeness denotes that each potential string has a unique suffix in $\ncalA$.
\end{Def}

\begin{Def}[Context Tree]
	A context tree is a set $\ncalA$ of suffixes, or \emph{nodes}, that together form a complete tree on the source alphabet $\ncalD = \ncalY \times \ncalW$. A node at depth $\ell$ will have the form $(\nby^{\ell},\nbw^{\ell-1})$. We assume that the depth of the tree is bounded by a constant $D \in \nbbN_{+}$. To each node $\nba_{i} \in \ncalA$, a probability estimate $\hat{P}$ and estimated cost-to-go value $\hat{J}$ is assigned. Alternatively, a context tree is a perfect tree of depth $D$.
\end{Def}

\begin{assum}
	The maximum depth of the context tree, $D$ is longer than the memory length of the radar scene's measurement generating process $U$. This assumption can be relaxed via an extension of context-tree weighting to infinite trees \cite{Willems1998}, but we leave it for simplicity and due to the fact that temporal correlation in a realistic radar scene will be finite.
\end{assum}

The idea of building such a tree-source to efficiently represent the stochastic behavior of a dynamic system was proposed for the general reinforcement learning problem in \cite{Farias2010}, using an active variant of the Lempel-Ziv algorithm. In this interpretation, each node of the tree source corresponds to strings of radar observations and waveforms $(\nby^{\ell},\nbw^{\ell-1})$, along with an estimated transition probability and estimated cost-to-go. Thus, the key components of a solution are the estimates $\hat{P}$ and $\hat{J}$.\\

In our approach, we apply the general Lempel-Ziv inspired framework, using the more practical context-tree weighting method of Willems \emph{et al.} \cite{Willems1995}, which improves the rate at which the transition probabilities are estimated. The context tree weighting method is particularly useful as it achieves the optiumum worst-case average redundancy for any stationary and ergodic tree source \cite{Willems1995}. While performance guarantees for LZ algorithms are asymptotic and in the average sense, the CTW algorithm performs well on any sequence in a given model class. The details of the CTW algorithm used for waveform selection will now be discussed, and can be seen in Algorithm \ref{algo:arlz}. \\

\begin{algorithm}[t]
	\setlength{\textfloatsep}{0pt}
	\label{algo:arlz}
	\caption{Adaptive Radar via Context-Tree Weighting}
	\SetAlgoLined
	\textbf{Input} discount factor $\gamma$ and a sequence of exploration rates $\{\alpha_{k}\}$\\
	Set $c \leftarrow 1, \tau_{c} \leftarrow 1$, $N(\cdot) \leftarrow 0 \; \; \textit{(context counts)}, \hat{J}(\cdot) \leftarrow 0$, $\hat{P}(\cdot) \leftarrow 1 / \lvert \mathcal{Y} \rvert$  \\
	\For{\text{Each CPI}}{
		\vspace{0.07cm}
		Radar observes $\nby_{k} \in \mathcal{Y}$; \textit{(Most recent measurement)}\\
		\eIf{$N(\nby_{\tau_{c}}^{k},\nbw_{\tau_{c}}^{k-1}) > 0$ (\textit{known context})}{ 
			Select a random waveform $\nbw_{k} \in \mathcal{W}$ with probability $\alpha_{k}$; (\textit{Exploration})\\
			OR select greedy waveform with respect to $\hat{J}$ with probability $1-\alpha_{k}$; (\textit{Exploitation})	
		}{
			Select $\nbw_{k}$ uniformly from $\mathcal{W}$; \textit{(Exploration in an unknown context)}\\
			\For{$u = k:-1:\tau_{c}$ (\textit{Traverse backwards and perform updates}) }{ 
				Increment $N(\nby_{\tau_{c}}^{u}, \nbw_{\tau_{c}}^{u}) \leftarrow N(\nby_{\tau_{c}}^{u}, \nbw_{\tau_{c}}^{u}) + 1$; (\textit{Update node count or add node to tree}) \vspace{0.3cm}
				
				For each $\nby_{u} \in \mathcal{Y}$ update $\hat{P}(\nby_{u} | \nby_{\tau_{c}}^{s-1}, \nbw_{\tau_{c}}^{s-1}) \leftarrow \frac{\textstyle N( (\nby_{\tau_{c}}^{u-1},\nby_{u}), \nbw_{\tau_{c}}^{u-1}) + 1 / 2}{\textstyle \sum_{y'} N((\nby_{\tau_{c}}^{u-1},\nby'),\nbw^{u-1}_{\tau_{c}}) + \lvert \mathcal{Y} \rvert}$; (\textit{KT Update})\vspace{0.3cm}
				
				If node is not a leaf node, weigh $\hat{P}(\nby_{u}|\nby_{\tau_{c}}^{s-1}, \nbw_{\tau_{c}}^{s-1})$ by its children. Otherwise, use KT estimate above. (\textit{Context Tree Weighting (\ref{eq:ctw})})\vspace{0.3cm}
				
				Update cost-to-go $\hat{J}(\nby_{\tau_{c}}^{u}, \nbw_{\tau_{c}}^{u-1}) \leftarrow \underset{\nbw_{u}}{\min} \textstyle \sum_{\nby_{u+1}} \mathbb{P}(\nby_{u+1}|\nby_{\tau_{c}}^{u}, \nbw_{\tau_{c}}^{u-1}) \times [g(\nby_{u},\nbw_{u},\nby_{u+1}) + \gamma J(\nby_{\tau_{c}}^{u}, \nbw_{\tau_{c}}^{u-1})];$
				
			}
			$b \leftarrow b+1$, $\tau_{c} \leftarrow \tau_{c}+1$; (\textit{Begin the next phrase})
		}
		Radar receives cost $g(\nby_{k},\nbw_{k},\nby_{k+1})$;
	}
\end{algorithm}

To build the context tree, the algorithm splits time into \emph{phrases} of variable length. Each phrase consists of a sequence of observations and actions, and corresponds to a node of the context-tree. Let the $c^{\text{th}}$ phrase cover the interval $\tau_{c} \leq k \leq  \tau_{c+\ell}-1$. The associated sequence of measurements and waveforms which characterizes the phrase will then be $(\nby_{\tau_{c}}^{\tau_{c+\ell}}, \nbw_{\tau_{c}}^{\tau_{c+\ell}-1})$, which is used to estimate transition probabilities and cost-to-go values. The sequence $(\nby_{\tau_{c}}^{\tau_{c+\ell}}, \nbw_{\tau_{c}}^{\tau_{c+\ell}-1})$ along with the associated transition probability and cost-to-go estimates, corresponds to a node of the context tree. \\

For each pair $\nby_{\ell+1} \in \mathcal{Y}$ and $\nbw_{\ell} \in \mathcal{W}$, the algorithm maintains an estimate of the transition behavior $\hat{P}(\nby_{\ell+1}|\nby^{\ell},\nbw^{\ell})$, which is the probability of observing a particular measurement of the radar scene $\nby_{\ell+1}$ at the next time step when waveform $\nbw_{\ell}$ is selected, given the current context $(\nby^{\ell},\nbw^{\ell-1})$. The transition probabilities are initialized to a uniform distribution over the measurement space $\mathcal{Y}$ and updated using the observed counts of particular contexts. Let $N(\nby^{\ell+1},\nbw^{\ell})$ be the number of times the context $(\nby^{\ell},\nbw^{\ell})$ has been visited before step $k$. Then the transition probability can be estimated using the Krischevsky-Trofimov estimator
\begin{equation}
	\label{eq:kt}
	P_{\texttt{KT}}(\nby_{\ell+1}| \nby^{\ell},\nbw^{\ell}) = \frac{N(\nby^{\ell+1},\nbw^{\ell})+1/2}{\sum_{\nby' \in \mathcal{Y}} N((\nby^{\ell},\nby'),\nbw^{\ell})+|\mathcal{Y}|/2},
\end{equation}
which can be computed sequentially using the observed frequency of each context at the current tree level. The KT estimator can be viewed from a Bayesian perspective as a Dirichlet-Multinomial conjugate pair with an uninformative prior. The KT estimator has several important properties, including the following Lemma, for which a proof is available in \cite{Merhav1998}.

\begin{lemma}[Redundancy Performance of KT Estimator]
	\label{thm:kt}
	Given an arbitrary sequence $\{y_{t}\},$ with each $y_{t}$ in a finite alphabet $\ncalY$, consider the problem of making sequential probability assignments $\hat{P}(\cdot)$ over $\ncalY$, so as to minimize the compression loss, or \emph{self-information}, $\sum_{t=1}^{T} -\log \hat{P}_{t-1}(y_{t})$ for a finite time horizon $T$. Then, the KT estimator (\ref{eq:kt}) \textbf{uniformly} achieves the bound
	\begin{equation}
		-\sum_{t=1}^{T} \log \hat{P}_{\texttt{KT}}(y_{t}) - \min_{q \in \ncalM(\ncalY)} \left[ -\sum_{t=1}^{T} \log q(y_{t}) \right] \leq \frac{\ncalY}{2} \log T + O(1),
	\end{equation}
	where the minimization is taken over $\ncalM(\ncalY)$, the set of all probability distributions on $\ncalY$. As $T \rightarrow \infty$, the KT estimator performs close to the best probability assignment while retaining the simple update (\ref{eq:kt}).
\end{lemma}

\begin{remark}
	Since Lemma \ref{thm:kt} gives a bound on the performance of the KT estimator relative to the best constant probability assignment, made with full knowledge of the state sequence, it can be shown that acting greedily with respect to the KT estimated transition probabilities is asymptotically optimal, (see Theorem 1 of \cite{Farias2010}).
\end{remark}

\begin{Def}[Weighted probability estimation]
	Let $\Sigma$ be a finite context alphabet, $P^{\sigma}_{\texttt{KT}}$ be the KT estimate of node $\sigma \in \Sigma$, and $K$ be the current depth of the context tree model. The context tree weighting strategy then develops a weighted estimate of $P$, which is recursively defined as
	\begin{equation}
		P_{w}^{s}\left(\boldsymbol{\sigma}^{t-1}\right):= \begin{cases}\frac{1}{2} P_{\texttt{KT}}^{s}\left(\boldsymbol{\sigma}^{t-1}\right)+\frac{1}{2} \prod_{\sigma \in \Sigma} P_{w}^{\sigma s}\left(\boldsymbol{\sigma}^{t-1}\right) & \text { if }|s|<K \\ P_{\texttt{KT}}^{s}\left(\boldsymbol{\sigma}^{t-1}\right) & \text { if }|s|=K\end{cases},
		\label{eq:ctw}
	\end{equation}
	which corresponds to a weighting of both the estimated probability of a node occurring and the product of weighted probabilities corresponding to its \emph{children}.
\end{Def}

\begin{remark}
	The weighted probability estimate defined in Definition \ref{eq:ctw} is not unique. However, it is simple, allowing for redundancy analysis, as in \cite[Section VI. B]{Willems1995}, and represents a weighting over the estimated probabilities corresponding to all possible models which reside above the current node in the context tree. Under the weighting of (\ref{eq:ctw}), the context-tree weighting strategy can be shown to achieve the optimal worst-case average redundancy for finite-state sources.
\end{remark}

\begin{remark}
	The storage complexity of context-tree weighting is linear in the sequence length.
\end{remark}

In addition to the transition probabilities, the \emph{cost-to-go} function must be estimated for each context. The estimated cost $\hat{J}(\nby^{\ell+1},\nbw^{\ell})$ is initialized to zero and subsequently updated using the rule
\begin{equation}
	\hat{J}(\nby_{\tau_{c}}^{s}, \nbw_{\tau_{c}}^{s-1}) \leftarrow \underset{w_{s}}{\min} \sum_{\nby_{s+1}} \mathbb{P}(\nby_{s+1}|\nby_{\tau_{c}}^{s}, \nbw_{\tau_{c}}^{s-1}) \times [g(\nby_{s},\nbw_{s},\nby_{s+1}) + \gamma J(\nby_{\tau_{c}}^{s}, \nbw_{\tau_{c}}^{s-1})],
\end{equation}
where $\gamma$ is a weighting term for prior estimates called the \emph{discount factor} and the update is performed by traversing backwards over the outcomes which have been previously observed, and the transition probabilities are estimated using (\ref{eq:kt}). During each step, the action is selected with the intent of either exploiting the action which is known to be most effective or gathering information about under-explored actions. This behavior is controlled by the sequence of exploration probabilities $\{\alpha_{k}\}$.\\

\subsection{Asymptotic Analysis}
We now show that the proposed algorithm converges to the optimal long-term average cost. This analysis follows from the favorable redundancy properties of the CTW probability estimates, and some results from the theory of exact dynamic programming with estimated transition probabilities \cite{Mastin2012}. The basic idea is to define criteria for transition probability and cost-to-go estimates for which acting according to our estimated values will be equivalent to acting on $P$ and $J$ respectively. We use arguments based on Ziv's inequality and the total variation distance of probability estimates to show that the algorithm takes suboptimal actions only a small fraction of the time while exploiting. By decaying the sequence of exploration probabilities fast enough, this ensures average cost optimality.\\

\begin{lemma}
	 Let $\tilde{P}: \ncalY^{K} \times \ncalW^{K} \mapsto [0,1]$ and $\tilde{J}: \ncalY^{K} \times \ncalW^{K-1} \mapsto \nbbR$ be arbitrary transition kernel and value functions. Then, given a discount factor near one, there exists a constant $\varepsilon > 0$ such that when conditions
	 \begin{align}
	 	    \label{eq:pEst}
	 		\left\| \tilde{P}\left(\cdot \mid \nby^{K}, \nbw^{K}\right)-P\left(\cdot \mid \nby^{K}, \nbw^{K}\right)\right\|_{1} &< \varepsilon, \quad \forall \nby^{K}, \nbw^{K} \in \ncalY^{K} \times \ncalW^{K} \\
	 		\left| \tilde{J}\left(\nby^{K}, \nbw^{K-1}\right)-J^{*}\left(\nby^{K}, \nbw^{K-1}\right)\right|&< \varepsilon, \quad \forall \nby^{K}, \nbw^{K-1} \in \ncalY^{K} \times \ncalW^{K-1},
	 \end{align}
 hold, acting greedily with respect to $(\tilde{P},\tilde{J})$ results in actions which are optimal with respect to $(P,J_{\gamma}^{*})$. 
\end{lemma}

\begin{proof}
	The existence of $\vE$ follows immediately from the fact that the dynamic programming operator is a contraction mapping \cite[Proposition 4.1]{Bertsekas2006} and the finiteness of the state-action space. The optimality of acting on these estimates follows from the assumption of bounded cost per stage \cite[Theorem 2]{Mastin2012}.
\end{proof}

\begin{Def}[Total Variation Distance]
	For two probability distributions $p(x)$ and $q(x)$ over common alphabet $\mathcal{X}$, the total variation distance can be defined as
	\begin{equation}
		\operatorname{TV}(p,q) = \frac{1}{2} \sum_{x}|p(x)-q(x)|,
	\end{equation}
	which satisfies the properties of a statistical metric, and is related to KL divergence by Pinsker's inequality $\operatorname{TV}(p,q) \leq \sqrt{ \frac{1}{2} \kl(P||Q)}$, which we will use later on. We note that the condition (\ref{eq:pEst}) is equivalent to a total variation distance no greater than $\vE/2$.
\end{Def}

\begin{Def}[$\bar{\varepsilon}$ inaccuracy]
	Define $\mathcal{I}^{\bar{\varepsilon}}_{t}$ to be the event that at time $t$ at least one of the following conditions hold:
	\begin{enumerate}
		\item $\operatorname{TV}(\hat{P}_{t}(\cdot|\nby_{\tau_{c(t)}}^{t}, \nbw_{\tau_{c(t)}}^{t}),P(\cdot|\nbw_{t-K+1}^{t},\nbw_{t-K+1}^{t})) > \bar{\vE}$
		\item The current context has not been visited prior to $t$.
	\end{enumerate}
\end{Def}

In \cite{Farias2010}, it was shown that if the exploration rate decays slow enough, then controlling the fraction of time the algorithm is $\bar{\varepsilon}$ inaccurate is sufficient to ensure that the fraction of time the algorithm is $\bar{\varepsilon},\bar{K}$ inaccurate goes to zero asymptotically as well. Further, it was shown in Theorem 2 of \cite{Farias2010} that vanishingly low $\bar{\varepsilon},\bar{K}$ inaccuracy is enough to ensure asymptotic average cost optimality for any transition kernel.

\begin{Def}[$\bar{\varepsilon},\bar{K}$ inaccuracy]
	Define $\mathcal{B}_{t}^{\bar{\varepsilon},\bar{K}}$ to be the event that at time $t \geq K$ either of the following conditions hold:
	\begin{enumerate}
		\item The length $d(t)$ of the current context is less than $K$.
		\item There exist $\ell$ and $(\nby^{\ell},\nbw^{\ell})$ such that any of the following hold:
		\begin{enumerate}
			\item $d(t) \leq \ell \leq d(t)+\bar{K}$.
			\item ($\nby^{\ell},\nbw^{\ell})$ contains the current context as a prefix.
			\item The estimated transition probabilities $\hat{P}(\cdot|\nby^{\ell},\nbw^{\ell})$ are more than $\bar{\varepsilon}$ inaccurate or the context $(\nby^{\ell},\nbw^{\ell-1})$ has not been visited prior to time $t$.
		\end{enumerate}
	\end{enumerate}
\end{Def}

\begin{theorem}[Asymptotic Optimality]
	Algorithm \ref{algo:arlz} attains the optimal long-term average cost $\lambda^{*}$ almost surely as $k \rightarrow \infty$.
\end{theorem}

\begin{proof}[]
	Using the properties of the KT estimator and Pinsker's inequality, we can provide a lower bound on the cumulative KL divergence between $P$ and $\hat{P}$, $\sum_{k=1}^{K} \Delta_{k}$ by the event of $\bar{\varepsilon}$ inaccuracy. \\
	
	
	
	We begin by defining
	\begin{equation}
		\Xi_{t} \triangleq \operatorname{TV} \left( P(\cdot|\nby_{t-K+1}^{t},\nbw_{t-K+1}^{t}),\hat{P}_{t}(\cdot|\nby_{\tau_{c(t)}}^{t},\nbw_{\tau_{c(t)}}^{t}) \right),
	\end{equation}
	and
	\begin{equation}
		\Delta_{t} \triangleq \kl \left(P(\cdot|\nby_{t-K+1}^{t},\nbw_{t-K+1}^{t}),\hat{P}_{t}(\cdot| \nby_{\tau_{c(t)}}^{t},\nbw_{\tau_{c(t)}}^{t} ) \right)
	\end{equation}
	
	\begin{align*}
		\frac{1}{T} \sum_{t=1}^{T} \Delta_{t} &\geq \frac{2 \bar{\vE}^{2}}{T} \sum_{t=1}^{T} \mathbbm{1}_{[\Delta_{t} \geq 2 \bar{\vE}^{2} ]} \\
											  &\geq \frac{2 \bar{\vE}^{2}}{T} \sum_{t=1}^{T} \mathbbm{1}_{[\Xi_{t} > \bar{\vE}]}
	\end{align*}

	Let $F_{t}$ be the event that the current context $(\nby_{\tau_{c(t)}}^{t},\nbw_{\tau_{c(t)}}^{t})$ has not been visited prior to time $t$. By Ziv's lemma \cite[Lemma 13.5.5]{Cover2006}, we can write
	\begin{equation*}
		\sum_{t=1}^{T} \mathbbm{1}_{F_{t}} = c(T) \leq \frac{C_{2}T}{\log(T)}
	\end{equation*}

	Combining these two inequalities and using the definition of $\bar{\vE}$ inaccuracy, we have
	\begin{align}
		\frac{1}{T} \sum_{t=K}^{T} \mathbbm{1}\{\ncalI_{t}^{\bar{\vE}}\} &\leq \frac{1}{T} \sum_{t=1}^{T} (\mathbbm{1}_{\Xi_{t}>\bar{\vE}} + \mathbbm{1}_{F_{t}}) \nonumber \\ 
																		&\leq \frac{1}{2 \bar{\vE} T} \sum_{t=1}^{T} \Delta_{t} + \frac{C_{2}}{\log(T)},
	\end{align}
	from which we see the dominating term is $\Delta_{t}$, the KL divergence between the estimated and true transition probability kernels. From the redundancy property of the Context-Tree weighting method, the probability that $\Delta_{t}$ is large enough to cause $\bar{\vE}$ inaccuracy can be bounded by a term which vanishes as $T \rightarrow \infty$. Thus, we can bound the probability that $\bar{\vE}$ inaccuracy occurs very often in the following form
	\begin{equation*}
	\mathbb{P} \left( \frac{1}{T}\sum_{t=K}^{T} \mathbbm{1}\{\ncalI_{t}^{\bar{\vE}}\} \geq K_{1} \frac{\log \log T}{\log T} \right)  \leq \exp( - K_{2} T),
	\end{equation*}
	where $K_{1}$ and $K_{2}$ are constants.\\

	By the first Borel-Cantelli lemma, the fraction of the time $\bar{\vE}$ inaccuracy occurs is vanishingly small as $T \rightarrow \infty$. Having established this, we can apply Lemma 4 of \cite{Farias2010}, from which we can conclude that given a sequence of exploration probabilities that is non-increasing and of sufficiently fast decay,
	\begin{equation}
		\lim_{T \rightarrow \infty} \frac{1}{T} \sum_{t=K}^{T} \mathbbm{1}\{\ncalB_{t}^{\bar{\vE},\bar{K}}\} = 0 \quad \textit{almost surely}.
		\label{eq:ekIn}
	\end{equation}

	Average cost optimality then follows quickly. We can define the event that the algorithm selects a suboptimal action as 
	\begin{equation*}
		\{\nbw_{t} \notin \ncalW^{*}_{t}\} \subset \ncalB_{t}^{\bar{\vE},\bar{K}} \cup \{\textrm{Exploration at time t} \},
	\end{equation*}
	which occurs a vanishingly low fraction of the time as $T \rightarrow \infty$, given (\ref{eq:ekIn}) as well as the earlier assumption that the sequence of exploration probabilities decays sufficiently quickly. We have previously assumed that the optimal average cost was independent of the initial state of the Markov process. We can define $G_{n}(\nby^{k},\nbw^{k})$ to be the average cost of an optimal policy over a finite time window $T_{\delta}$, starting at time $t = n$. We can pick $T_{\delta}$ large enough such that
	\begin{equation*}
		\E[\max |G_{n} - \lambda^{*}|] < \varepsilon,
	\end{equation*}
	where each of the random variables $\max |G_{n} - \lambda^{*}|$ are independent and identically distributed over the range of $n$. By the strong law of large numbers we then have
	\begin{equation*}
		\limsup_{T \rightarrow \infty} \left| \frac{1}{T} \sum_{t=1}^{T} (g(\nby_{t},\nbw_{t}) - \lambda^{*}) \right| \leq \varepsilon,
	\end{equation*}
	with probability one.	
\end{proof}

The universal learning approach is particularly effective due to its general structure. Algorithm \ref{algo:arlz} will asymptotically converge to the long-term cost optimal for any stationary MDP of finite order and a discount factor sufficiently large, $\gamma \approx 1$. Thus, the common MDP, contextual bandit, and multi-armed bandit sequential decision formulations are all contained within this umbrella, and a large class of problems can be optimally solved regardless of the structure of the underlying transition probabilities.\\

\section{Cost Function Design}
\label{se:cost}
Due to the large body of research on the statistical theory of radar detection and estimation, many rigorous performance measures can be utilized. Generally, the problem of performance feedback is approached from either a control-theoretic or information-theoretic perspective. The former involves direct feedback from the tracking system to improve \emph{system level} performance as opposed to measurement level performance, and optimizes quantities such as the mean square tracking error or size of the target validation gate in measurement space \cite{Kershaw1994}. The information-theoretic perspective generally aims to maximize mutual information between the target and received signal, as in \cite{Bell1993}. However, to obtain a closed form expression, simplifications are often necessary.\\

For example, consider the problem of minimizing squared tracking error. In most cases, it is not possible to evaluate the mean square error (MSE) matrix analytically \cite{Bell2015}. Thus, it is common to use the \emph{Bayesian Cram\`er-Rao lower bound} in place of the MSE matrix. For target tracking, this involves conditioning on the observed data and computing the \emph{predicted conditional Cram\`er-Rao lower bound} (PC-CRLB). The PC-CRLB consists of a prior term and a data term. Unfortunately, the data term is difficult to compute in general, and it is common to assume a Gaussian measurement model.\\

In the information theoretic viewpoint, the target's impulse response is assumed to be a random vector $\mathbf{g}(t)$. If the radar transmits waveform $x(t)$, the resulting scattered signal $\mathbf{z}(t)$ is a finite-energy random process given by the convolution of $\mathbf{g}(t)$ and $x(t)$. Thus, a reasonable goal is to find waveforms which maximize the mutual information $I(\mathbf{g}(t); \mathbf{y}(t))$, where $\mathbf{y}(t)$ is the sum of $\mathbf{z}(t)$ and an additive noise process. The conditional mutual information $I(\mathbf{g}(t);\mathbf{y}(t)|x(t))$ is then easily computed if $\mathbf{g}(t)$ is a Gaussian process and the additive noise is Gaussian and independent of the transmitted waveform and target. Under these restrictive assumptions, Bell \cite{Bell1993} develops an optimal waveform design algorithm, based on the information-theoretic idea of waterfilling. Unfortunately, the proposed approach requires prior knowledge of the variance of $\mathbf{g}(t)$.\\

Under both viewpoints, modeling assumptions are required for tractable analysis. We can instead consider similar approaches, where the distributions are \emph{learned} over time by considering the relationship between particular waveform/observation pairs and the associated cost. The first objective function utilized is will be referred to as the tracking innovation objective and is defined as follows.\\

\begin{Def}[Tracking Innovation]
	The tracking innovation objective function is given by 
	\begin{equation}
		g_{\texttt{track}} \triangleq (\mathbf{Z}_{k}-\mathbf{\hat{X}}_{k})^{2},
	\end{equation}
	where $\mathbf{Z}_{k}$ is the current unfiltered measurement vector containing a range and velocity estimate for the target at time step $k$, and $\mathbf{\hat{X}}_{k}$ is the most recent target state estimate given by the tracking filter.
\end{Def}

Additionally, we propose an information theoretic objective function which seeks to minimize the negative entropy in the delay-Doppler image. This objective is defined as follows.

\begin{Def}[Negative range-Doppler Entropy]
	The negative entropy objective function is given by
	\begin{equation}
		\label{eq:negentropy}
		g_{\texttt{entr.}} \triangleq \sum_{i=1}^{N} \sum_{j=1}^{M} p_{ij}(k) \log(p_{ij}(k)),
	\end{equation}
	where the probability mass function $p_{ij}(k)$ is the probability that the target is located at delay-Doppler coordinate $(i,j) \in \mathcal{H}$ given the entire sequence of measurements. In practice, the probability of a target being present can be established using approximations, as in \cite{LaScala2005}, or by setting a detection threshold, and calculating a normalized distance from the energy in each cell to the threshold to establish a probability of the target being present. In this work, we opt for the latter approach. If the energy in a particular cell is very far from the detection threshold, then there is little ambiguity. Thus, minimizing (\ref{eq:negentropy}) will reduce uncertainty about the target's position.
\end{Def}

\section{Waveform Library Design}
\label{se:library}
Thus far, we have mathematically treated the waveform selection process as a stochastic control problem, but have neglected many of the important considerations for practical radar systems. Of particular importance is the class of waveforms the radar will select from, given by $\mathcal{W}$ in our above formulation. This subject has been specifically treated in \cite{Cochran2009}, \cite{Howard2004}, \cite{Suvorova2006}. In these works, information theoretic criteria are proposed to measure the effectiveness of a particular set of waveforms. Generally, the goal is to compensate for lack of real-time knowledge about the system's current state by utilizing a waveform library that is rich enough such that a sufficiently good waveform is contained for every possible performance criterion and real-time information state that may be encountered. Cochran et al. suggest the following measure of effectiveness
\begin{equation}
	G_{f}(\mathcal{W}) = \int_{\mathbf{P} > 0} \max_{w \in \mathcal{W}} \log \det (\mathbf{I}+\mathbf{R}_{w}^{-1} \mathbf{P}) \rm{d}F( \mathbf{P} ),
\end{equation}
where $\mathbf{P} > 0$ is the positive semidefinite state covariance matrix. Thus, this measure represents an expected value over the distribution of all possible state covariance matrices. Unfortunately, this formulation requires knowledge of the the distribution of the state covariance matrix, which is specific to the particular radar scene being considered. While this information might be available from the tracking system, it is also not explicitly clear how to encode arbitrary information into this distribution in general. However, from the above formulation, some key design choices can be highlighted. For example, in \cite{Suvorova2006}, this concept is utilized to show that for a library of linear FM waveforms, all but the waveforms with the minimum and maximum chirp rates can be removed without a loss of redundancy.\\

In practice, a waveform library should contain enough waveforms such that the signals have sufficiently different ambiguity properties and frequency responses. Additionally, the waveforms considered must be easy to implement in a practical system. Finally, the waveform library should be kept sufficiently small, so that each waveform can be easily associated with particular contexts and the complexity of any associated decision process is kept manageable.\\

In this work, we focus on a waveform library containing three well-known radar waveforms, the Linear frequency modulated (LFM) chirp, the phase-coded pulse using a Zadoff-Chu sequence to determine the phase values, and a generalized nonlinear FM (NLFM) waveform.\\ 

The complex envelope of the LFM chirp is expressed as
\begin{equation}
	\tilde{s}(t) = \frac{1}{\sqrt{T}} \textrm{rect} \left(\frac{t}{T} \right) \exp(j\pi k t^{2}), \quad k = \pm \frac{B}{T},
\end{equation}
where $T$ is the pulse duration, $B$ is the bandwidth, and $\textrm{rect}(\cdot)$ is the rectangular function. LFM is the most popular pulse-compression method due to the favorable delay resolution and ease of implementation. However, LFM suffers from considerable range-Doppler coupling due to the shape of the associated AF. Additionally, relatively strong delay sidelobes can remain in the autocorrelation function. A detailed description of the properties of LFM waveforms can be found in Chapter 4 of \cite{Levanon2004}.\\

To provide diversity in the waveform ambiguity properties, we also consider phase-coded pulses, which are another common method for pulse compression. This approach breaks the pulse into $M$ bits of equal duration. Each bit is then coded with a different phase value. The complex envelope of a phase coded pulse is given by
\begin{equation}
	\tilde{s}(t) = \frac{1}{\sqrt{T}} \sum_{m=1}^{M} u_{m} \mathrm{rect} \left[ \frac{t-(m-1)-t_{b}}{t_{b}}\right],
\end{equation}
where $u_{m} = \exp(j \phi_{m})$ and the set of $M$ phases $\{\phi_{i}\}_{i=1}^{M}$ is the code associated with the waveform. Here, we use the well-known Zadoff-Chu code, which is applicable for any code length $M$. These sequences are known for their perfect auto-correlation properties, and are described in detail in Chapter 5 of \cite{Levanon2004}. \\

The final class of waveform we consider is a generalized NLFM chirp \cite{Papandreou1994,Sira2007,Sira2009}, having complex envelope
\begin{equation}
	\tilde{s}(t)=a(t) \exp \left(j 2 \pi b \xi\left(\frac{t}{t_{r}}\right)\right), \quad|t| \leq \frac{T}{2}+t_{f},
\end{equation}
where $a(t)$ is a trapezoidal envelope with rise/fall time $t_{f} \ll T/2$, $b \in \mathbb{R}$ is the FM rate parameter, $\xi(\frac{t}{t_{r}})$ is the chirp's phase function, and $t_{r} > 0$ is a reference point in time. By varying the phase function $\xi$, different FM waveforms, such as hyperbolic and exponential FM chirps \cite{Sira2007}. It has been observed that when the waveform's time-frequency characteristics are dynamically varied, nonlinear FM waveforms often provide better tracking performance than linear FM. One reason for this behavior is minimal range-Doppler coupling, which is desirable for many tracking scenarios. 

\section{Numerical Results}
In this section, we examine the performance of the universal learning approach in two classes of radar scenes. We compare the proposed algorithm to the optimal stationary policy, the active Lempel-Ziv algorithm proposed in \cite{Farias2010}, as well as the optimal policy for a first-order MDP formulation, such as in \cite{Thornton2020}. We first consider waveform-agile radar tracking in an adversarial scene, where the measurement process depends on the radar's previously transmitted waveforms. An example of a practical scene of this type would be in the presence of an adaptive emitter, which tracks the radar's behavior over an extended time interval and adjusts its transmission policy accordingly. We additionally consider a scene which evolves according to a higher-order Markov process, but does not depend on the radar's choice of waveforms. An example of this type of scene is a non-cooperative radar-communications coexistence scenario, or tracking in clutter with temporal correlations.\\

The simulations consist of 50 target tracks, where each track consists of a randomly initialized target moving at constant velocity for $200$ radar CPIs. A CPI consists of $128$ radar pulses, with a pulse repetition interval of $.4\textrm{ms}$. At the conclusion of each CPI, the radar makes an observation of the target's position and feeds the measurement into a tracking filter. To characterize the radar's performance, we examine the RMSE over the course of each target track, as well as the average measured \texttt{SINR} over the course of the target track.\\

\begin{figure*}[t]
	\centering
	\includegraphics[scale=0.6]{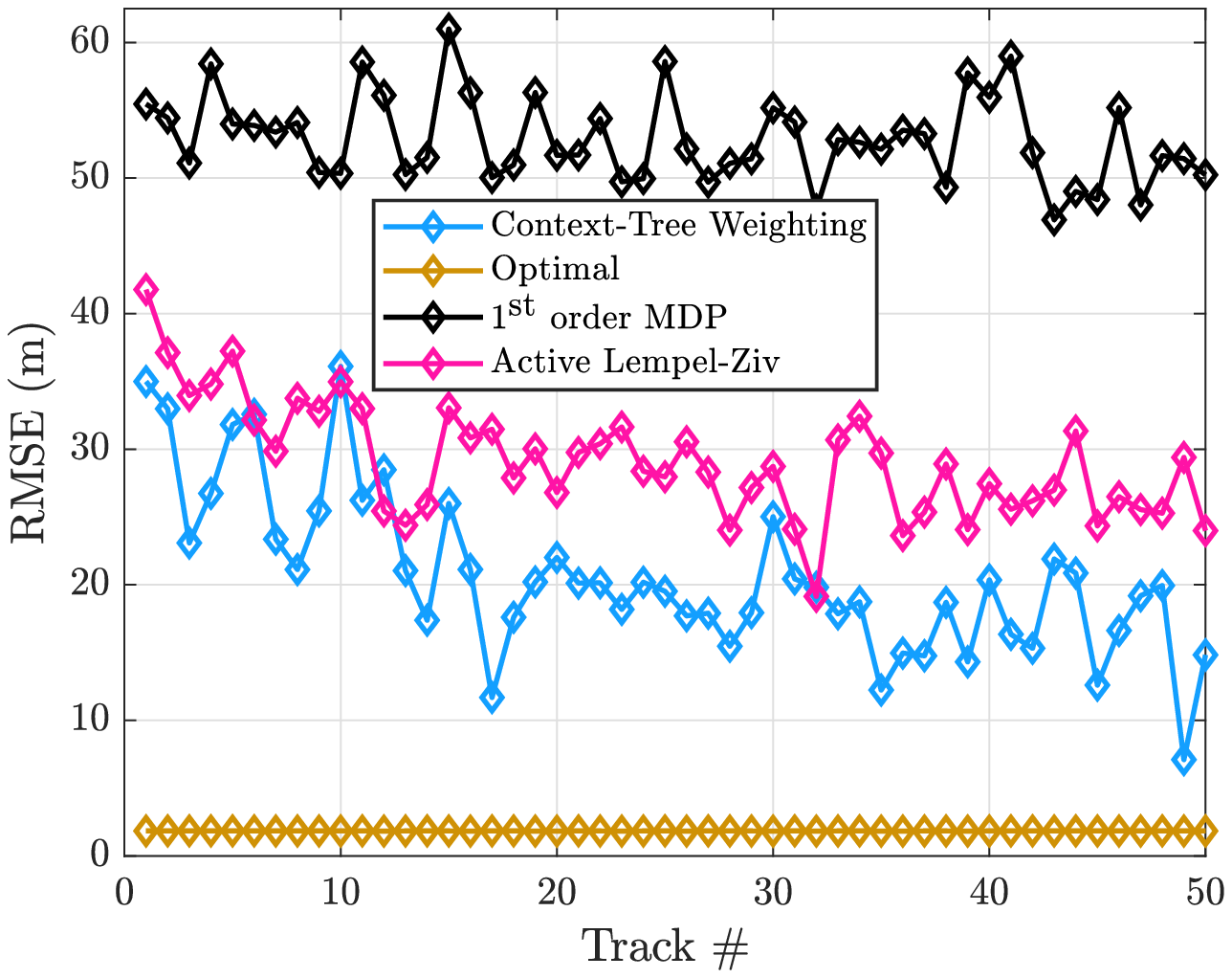}
	\includegraphics[scale=0.6]{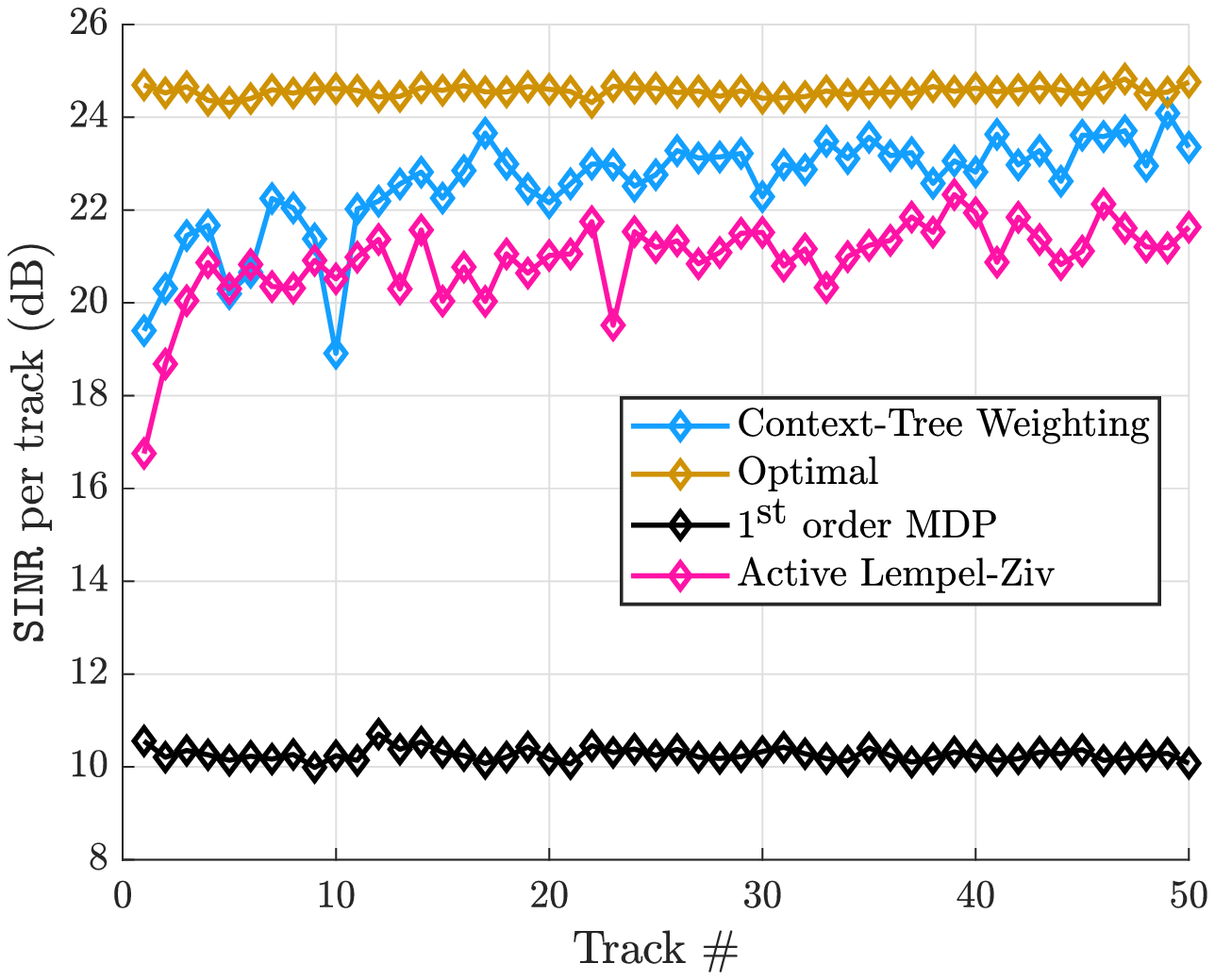}
	\caption{\textsc{Averge SINR} and \textsc{RMSE} per track in the presence of an adaptive emitter. We examine the radar's performance over a sequence of 50 target tracks, with each target track consisting of 200 radar measurements (CPIs). We observe that the universal learning algorithms tend towards the optimal policy and significantly outperform a first order MDP solution.}
	\label{fig:adversarial}
\end{figure*}

\begin{figure*}[t]
	\centering
	\includegraphics[scale=0.6]{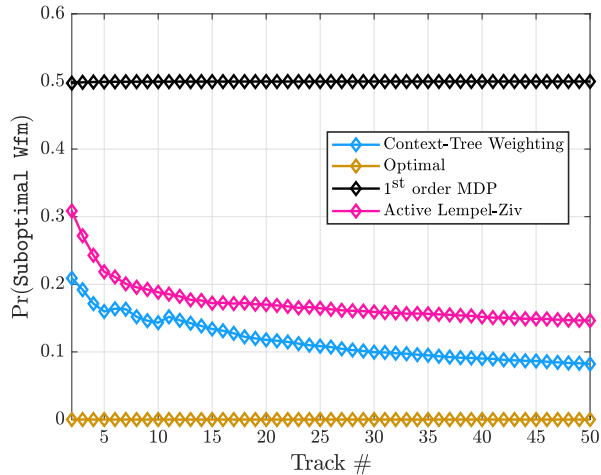}
	\caption{\textsc{Proportion of suboptimal waveforms} per target track, examined over the course of 50 tracks for each waveform selection algorithm in the non-cooperative scene, which is a Markov process of order $U = 3$. We note that reducing the number of actions in each context to one reduces the number of paths through the context-tree and improves convergence behavior.}
	\label{fig:convAdv}
\end{figure*}

In Figure \ref{fig:adversarial}, we observe the tracking performance of each waveform selection algorithm in the presence of an adaptive emitter. The emitter adapts it's state based on the radar's two most recently transmitted waveforms, as well as the emitter's previous state. Thus, the scene is a Markov process of order $U = 2$ that depends on the radar's actions. We see that the CTW algorithm begins to approach the optimal policy within the finite horizon. The performance of the active-LZ algorithm lags slightly behind the CTW algorithm, and typically performs around $2$-$3$dB worse than CTW.\\

Figure \ref{fig:convAdv} shows the proportion of suboptimal waveforms selected by each algorithm. We observe that the $1^{\textrm{st}}$ order MDP is unable to capture the target channel's higher-order dependence on the radar's past actions, and thus transmits sub-optimal waveforms at a fixed rate of $0.5$. As expected, the universal learning algorithms tend toward the optimal waveform selection policy as the sequence of target tracks proceeds. However, we note that the CTW algorithm immediately provides better performance than the active-LZ algorithm and maintains an advantage throughout the sequence. This can be attributed to the CTW algorithm using a weighted mixture of long context weights, whereas the active-LZ algorithm only uses information from the current level of the context-tree.\\

\begin{figure*}[t]
	\centering
	\includegraphics[scale=0.6]{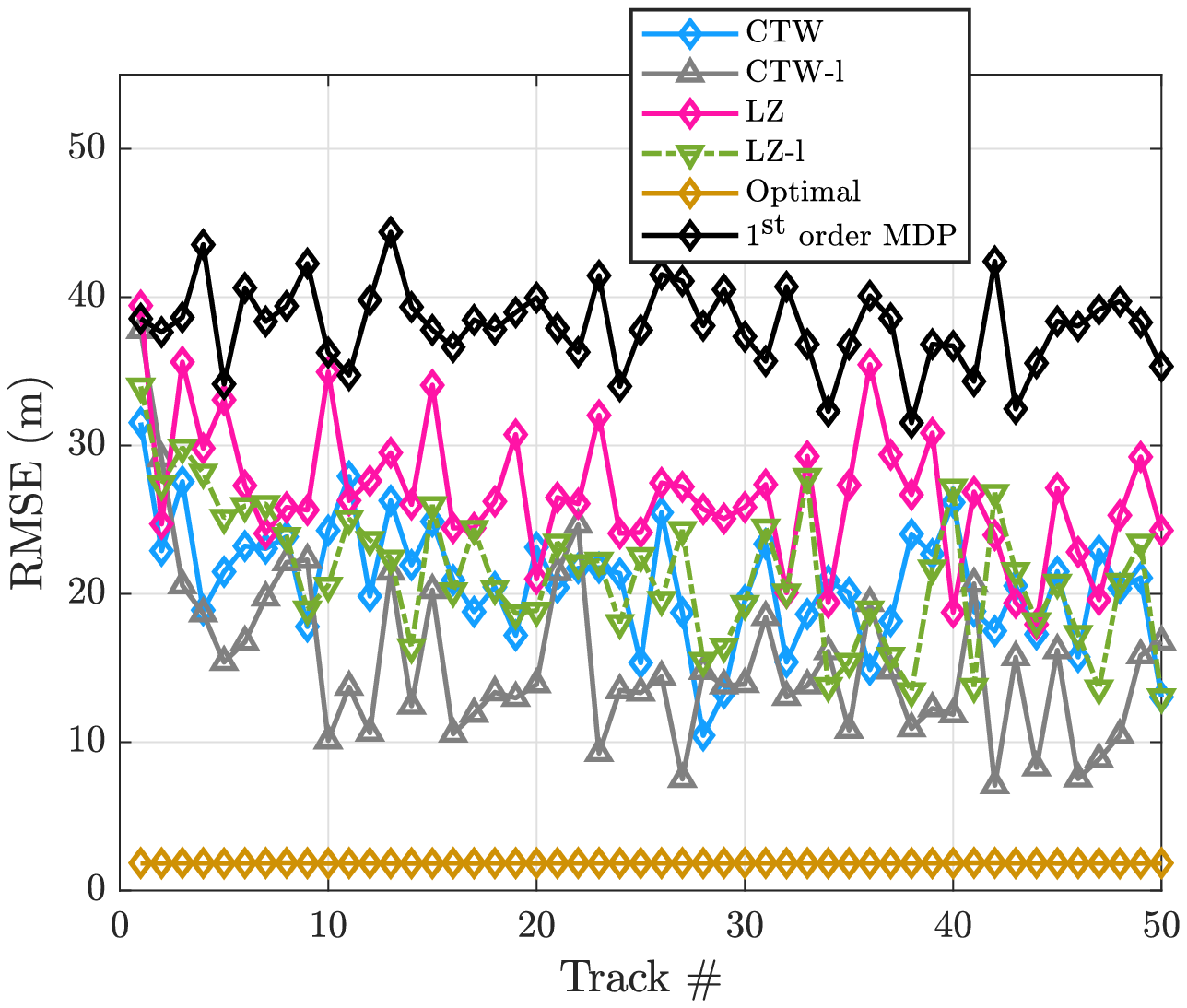}
	\includegraphics[scale=0.6]{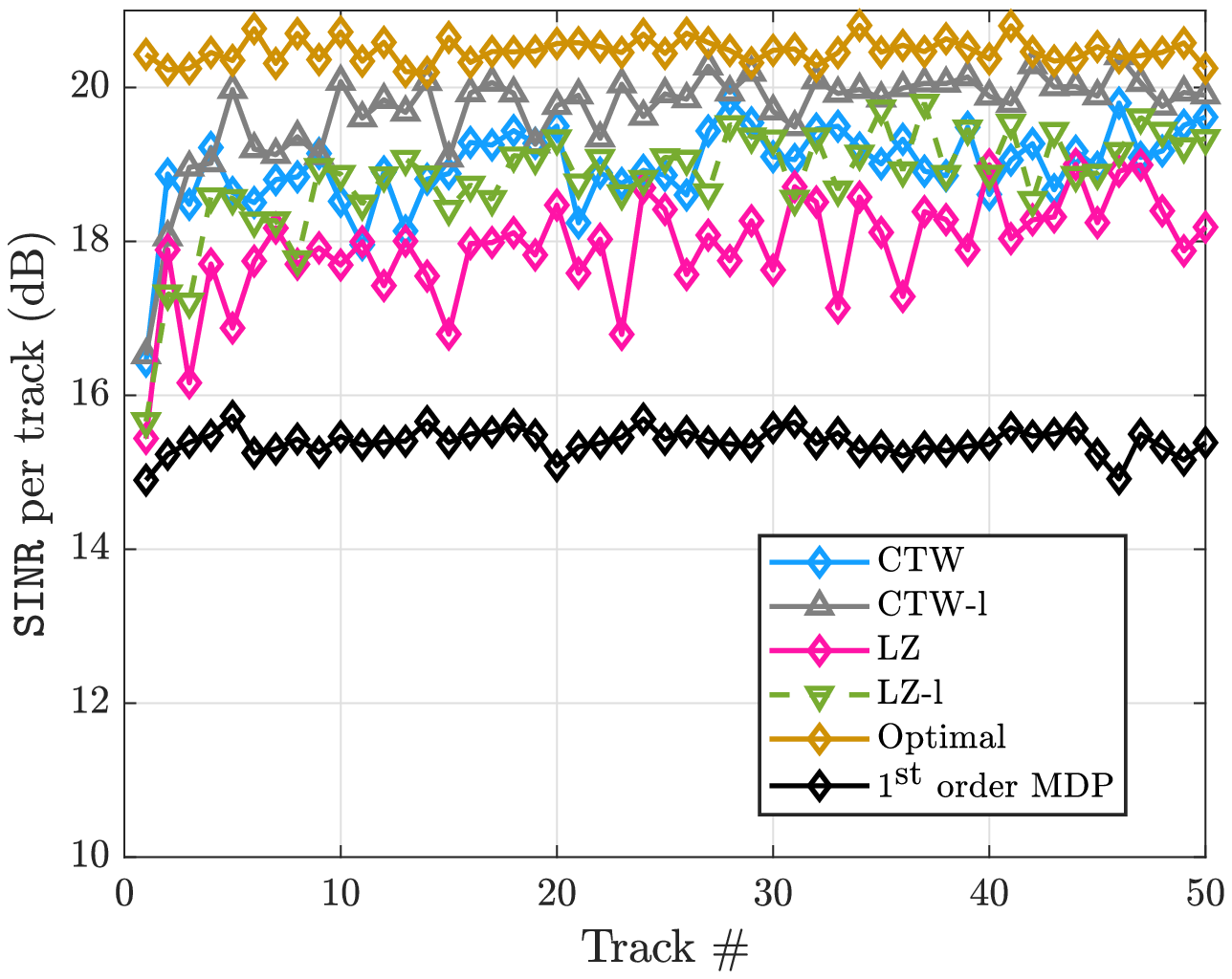}
	\caption{\textsc{Averge SINR} and \textsc{RMSE} per track in the presence of an adaptive emitter. We examine the radar's performance over a sequence of 50 target tracks, with each target track consisting of 200 radar measurements (CPIs). We observe that the universal learning algorithms }
	\label{fig:markov}
\end{figure*}

\begin{figure*}[t]
	\centering
	\includegraphics[scale=0.6]{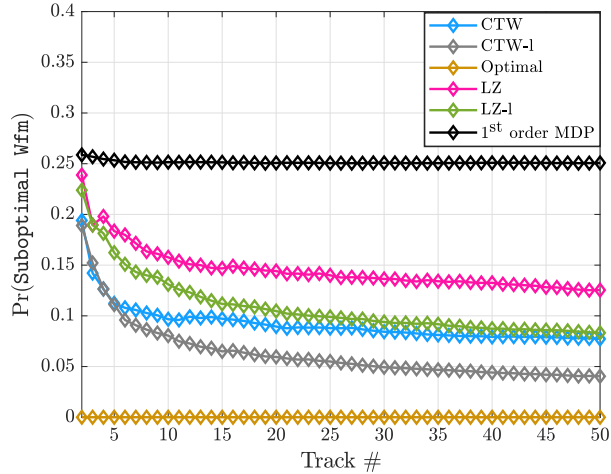}
	\caption{\textsc{Proportion of suboptimal waveforms} per target track, examined over the course of 50 tracks for each waveform selection algorithm in the non-cooperative scene, which is a Markov process of order $U= 3$. We note that reducing the number of actions in each context to one reduces the number of paths through the context-tree and improves convergence behavior.}
\end{figure*}

In Figure \ref{fig:markov}, we examine the tracking performance of each algorithm in the presence of a non-cooperative scene. The scene is a Markov process of order 3, which we compute using the procedure described in \cite{Raftery1985}. Since this scene's evolution does not depend on the radar's actions, we also examine the performance of variants of the CTW and active-LZ algorithms which consider contexts of the form $(\nby_{1}^{s},\nba_{s-1})$, which take into account the entire state sequence, but only the most recent action taken by the radar. Since many realistic scenes do not adapt to the behavior of the radar, this is a reasonable modification for many practical scenarios. This formulation greatly reduces the number of sample paths in the context-tree model and is expected to improve tracking performance since the exploration window can be greatly reduced.\\

We observe that limiting the dependence on previous actions does result in favorable performance in the finite time horizon examined here. In particular, the limited CTW algorithm (labeled as CTW-l in Figure \ref{fig:markov}) performs only slightly worse than the optimal policy in terms of average \texttt{SINR} per track. While the distance from the optimal policy is more significant in terms of RMSE, we observe that the general trend once again holds. 

\section{Conclusion and Future Directions}
\label{se:conclusion}
We have studied an asymptotically optimal stochastic control approach for radar waveform selection in any target channel that can be represented as a stationary Markov process of finite order. To jointly estimate the transition behavior and memory length of the environment, we employed an online learning approach using tools from universal source coding and prediction, namely the context-tree weighting algorithm. In particular, we applied a multi-alphabet version of the Context Tree Weighting strategy to represent the learning environment with optimal average redundancy across all tree sources.\\

This work has extended the previous literature on cognitive radar by considering variable order memory encoding for a general class of scenes. We have shown that this formulation provides additional generality, and results in performance improvements when compared to a $1^{\textrm{st}}$ order MDP. Importantly, our proposed algorithm inherits many of the desirable asymptotic properties of universal source coding and prediction, while making minimal assumptions regarding the radar scene itself. We have also demonstrated how prior knowledge can be used to guide the learning process by reducing the number of sample paths, or size of the probabilistic model, when the environment does not depend on a long history of the radar's actions.\\

While the results obtained for a simple adversarial scenario are promising, there are several key limitations of the current approach. First, while our algorithm has a quicker rate of convergence than the active-LZ algorithm described in \cite{Farias2010}, sample efficiency is still a major concern for active target tracking applications. In many practical cases, it may be more beneficial to constrain the model class and improve the convergence of the algorithm. Approaches such as meta-learning or transfer meta-learning may be appropriate methods for ensuring sample-efficiency over short horizon target tracks by capturing information about task similarity \cite{Thornton2021a}.\\

There are several areas for continued investigation. Even though the waveform selection problem has been well-studied, both computational and conceptual challenges remain. First, a deeper investigation of how prior information about target or scene behavior could be used to reduce the number of sample paths in universal learning algorithms would be of practical value, due to the sample-efficiency concerns mentioned above. Additionally, this work has considered the computation of a Bellman-optimal policy, which optimizes the discounted sum of long-term rewards. Future work could focus on optimizing for worst-case performance metrics, which may be of value in target tracking problems, where the goal is often to minimize the probability of a track being lost. Finally, there is a trade-off between the information a decision-maker requires to learn a specific objective and the resulting sub-optimality of a computed policy could be studied. In many cases, it would be preferable to select an objective with a slight sub-optimality cost in order to experience significantly improved convergence behavior.\\

\bibliographystyle{IEEEtran}
\bibliography{trackBib}{}

\end{document}